\documentclass[12pt]{llncs}
\usepackage{fullpage,epsfig}
{\bfseries}{\itshape}
\newtheorem{prop}{Property}{\bfseries}{\itshape}

\usepackage{amsfonts}

\def \P {{\mathcal P}}

\def \Rb {{\mathcal R}}
\def \X {{\mathcal X}}
\def \PR {{\mathcal P}_{\Rb}}

\def \Pr {{\mathbb P}}

\newcommand{\REMOVED}[1]{}

\begin{document}

\title{\textbf{Random input helps searching predecessors}
}

\author{D. Belazzougui\inst{1} \and A. C. Kaporis\inst{2} \and P. G. Spirakis \inst{3}
}

\institute{LIAFA, Univ. Paris Diderot - Paris 7, 75205 Paris Cedex 13, France. Email:\texttt{dbelaz@liafa.jussieu.fr}
\and
Department of Information and
Communication Systems Engineering, University of the Aegean,
Karlovassi, Samos, 83200. Email:\texttt{kaporisa@gmail.gr}
\and
Department of Computer Engineering
and Informatics, University of Patras, 26500 Patras, Greece; and
R.A.~Computer Technology Institute, N.~Kazantzaki Str, Patras
University Campus, 26500 Patras, Greece. Email:
\texttt{spirakis@cti.upatras.gr}
}


\setcounter{page}{1}
\pagestyle{empty}
\pagestyle{plain}
\maketitle


\begin{abstract}
We solve the dynamic Predecessor Problem with high probability (whp) in
constant time, using only $n^{1+\delta}$ bits of memory, for
any constant $\delta > 0$. The input keys are random  wrt a wider class of the
well studied and practically important class of $(f_1, f_2)$-smooth
distributions introduced in \cite{and:mat}.
It achieves $O(1)$  whp amortized  time. Its worst-case time is $O(\sqrt{\frac{\log n}{\log \log n}})$.
Also, we prove whp $O(\log \log \log n)$ time using only $n^{1+ \frac{1}{\log \log n}}= n^{1+o(1)}$ bits.
Finally, we show whp $O(\log \log n)$ time using $O(n)$ space.

\smallskip

\end{abstract}
%
%
\section{Introduction}
\label{Sec:Intro}
\vskip 0.05 cm
\noindent
{\bf  The problem.} Suppose that $X$ is a dynamic\footnote{It
supports updates as insertions of new and deletions of existing
keys, otherwise $X$ is static.} file of $n$ keys each represented by
$\ell \leq b$ bits of memory, drawn from a universe $U$ consisting
of $|U|= 2^{b}$ possible words, in a unit-cost RAM with word length
$b= \log |U|$. These $\ell$ bits are needed for representing  all
$n$ keys of $X$, so, $\ell$ must satisfy $\log |X|= \log n \leq \ell
\leq b = \log |U|$.
The goal is to wisely (preprocess $X$) and efficiently store  (total
bits for the $n$ keys of file $X$ to be as close to  $n$, while the
$\ell$ bits per key in $X$ to be as close to the minimum amount
$\log n$) file $X$ in a data structure in a way that important
queries are performed in the most efficient query time (as close to
$O(1)$). ``Time'' refers to the number of probed (compared) memory
words until the query is complete.
Perhaps the most basic is the {\em Membership} ${\tt Memb} (x, X)$
query: determine if $x \in X$, for an arbitrary $x \in U$.
The ${\tt Del}(y, X)$ query: delete $y \in X$ is reduced to ${\tt
Memb} (y, X)$ in $O(1)$ time.
Extensively studied \cite{BF02} and more complex to perform is the
{\em Predecessor} ${\tt Pred} (x, X)$ query: determine the largest
$y \in X$ such that $y \leq x$,  for an arbitrary $x \in U$.
Finally,  the ${\tt Insrt}(y, X)$ query inserts $y \in U$ into file
$X$.
\vskip 0.05 cm
\noindent
{\bf Complexity.}
\label{Par:Complexity}
%
%
It is known that membership queries on $X$ require $O(1)$ time and
linear space $O(n)$ \cite{FKS84} and \cite{Pag00} via exploiting the
power of double hashing.
However, searching for predecessors is substantially more complex,
as manifested by an early ground breaking result of Ajtai
\cite{Ajt88} proving the impossibility of achieving constant query
time using only polynomial wrt $n= |X|$ space. In more words, Ajtai
proved that if the memory bits per word are $\ell= O(\log n)$, that
is, only ``short'' words are available, and if only $n^{O(1)}$ such
short words  can be used for representing a set $|X|= n$, then, it
is impossible to perform predecessor search in constant time.
In sharp contrast, if ``long'' words of length $\ell > n$ are
available, Ajtai, Fredman, and Komlos \cite{AFK84} proved that only
$O(n)$ ``'long' words suffice for finding predecessors in constant
time on $|X|= n$.
A flurry of subsequent papers \cite{Mil94,BF02,PatTho06,PatTho07}
were motivated and built on the work of \cite{Ajt88}.
The result \cite{AFK84} when combined with \cite{Ajt88}, illustrates
how crucially limits the query performance the total available space
when measured by the overall bits.
Thus, it is important to assume that the $\ell$ bits per key of file
$X$ are very small wrt file size $n$, otherwise the space
requirement would be huge.
Quoting from \cite[Sect. 2]{BF02} ``...{\em The most interesting
data structures are those that work for an arbitrary universe whose
elements can fit in a single word of memory (i.e. $|U| \leq 2^b$)
and use a number of words that is polynomial in $n$, or ideally
$O(n)$} ...''.
%
%
%
%
%
\vskip 0.05 cm
\noindent
{\bf Motivation.} A key ingredient of the approaches in
\cite{Ajt88,Mil94,BF02,PatTho06,PatTho07} was to exhibit a ``bad''
distribution of inputs that, when combined with an also ``bad''
distribution of queries, ``ill advised''  the data structure to high
query time. Notably, such bad performance was unavoidable even by
allowing the (randomized) data structure to take random decisions
\cite{PatTho07} of rejecting excessively/unconveniently long
queries.
But, at this point an important question, that highly motivates our
work, arises: is it possible to circumvent these lower bounds by
just considering more ``natural'' and as general as possible real
world  distributions of input keys?
As Frieze and Reed  suggest \cite{FriRee98}, instead of analyzing the worst-case (but rare) performance of algorithms, it is often more telling to analyze their performance on typical (highly probable) instances.
Clearly, a natural goal is the algorithm to remain fast on the overwhelming
majority of random inputs.
Perhaps, the idea of exposing information from {\em uniform} input
distribution as to facilitate the performance of a data structure
goes back at least to Peterson's  method of {\em Interpolation
Search} (${\tt IS}$) \cite{Pet57}: recursively drive each probe for locating
${\tt Pred}(y,X)$ towards to (an estimation of) the expected
location of $y$ in file $X$.
He announced $O(n)$ such real (but ``long'') words suffice to find predecessors in order
of $\log \log n$ expected time.
Even for the restricted case of uniform (or, for transforming to uniform an arbitrary but {\em known}) continuous distribution, the analysis of ${\tt IS}$ attracted a lot of interest \cite{Gon77,GRG80,PerRei77,PIA78,Pri71,YY76}.
In Knuth's infamous list of top searching
problems, it was placed as 1st the ``{\em Average case of ${\tt IS}$}'',
emphasizing the importance of the typical (most likely)
data structure performance on (most general) random input
distributions, see Sedgewick's list \cite{Sed97}.
However, uniform input reals is a quite restrictive (strong)
assumption for real world applications, since these are not easy to
produce or handle in all practical applications. Willard \cite{W85}
made a breakthrough by introducing unknown {\em regular} input distributions
on reals (once more). In words, the probability density can
fluctuate over a given range, but, its first derivative remains
bounded, thus, bounding the estimated location error.
Mehlhorn and Tsakalidis \cite{meh:tsa} studied the dynamic version
for the larger class: {\em unknown} {\em $(\lceil n^{\alpha}\rceil, \sqrt{n})$-smooth} continuous distributions,
permitting more control to the distribution of the real keys in each
real subinterval. Intuitively, now 1st derivative may be unbounded. This class contains the class regular distributions, notably regulars correspond to $(\lceil n^{\alpha}\rceil, O(1))$-smooth ones. Finally, the state of the art of the most general unknown input distributions are the $(f_1, f_2)$-smooth ones (Def. \ref{Def:smooth-real}, in App. \ref{app:A0}) considered by Andersson and
Mattson \cite{and:mat}, for even more general $f_1, f_2$ parameters.
Recent advances in database theory and applications \cite{Gra06}
indicate the importance of ${\tt IS}$ variants on boosting the search
performance on the most general random (skewed) input distributions.

However, all these and many other ${\tt IS}$ based approaches (variants)
\cite{and:mat,Gon77,GRG80,KMSTTZ03,meh:tsa,PIA78,PerRei77,Pet57,W85,YY76}
raised the following two crucial concerns about their {\em space}
and {\em time} efficiency.  First, concerning space, in order to
facilitate their probabilistic analysis, all these approaches
\cite{and:mat,Gon77,GRG80,KMSTTZ03,meh:tsa,PIA78,PerRei77,Pet57,W85,YY76}
have a common suffering assumption. They are limited to random {\em
real} keys of {\em infinite} word length, thus, it is questionable
how to bound the total bits of space by a polynomial of $n= |X|$.
Frankly, it is not obvious to store $n$ real numbers in just $O(n)$
bits while achieving fast query time.
More on this, although these approaches require $O(n)$ real (but
``long'') words of memory, it is quite telling to recall our
discussion in the previous paragraph, about how the predecessor time
is boosted to a constant by just conveniently assuming ``long''
words of memory in \cite{AFK84}, opposed to the corresponding
superconstant lower bound proved in \cite{Ajt88} by assuming
``short'' ones.
Also, real keys do have theoretical significance, but, in all
applications it is highly not trivial to produce or handle.
Finally, concerning running time, these approaches
\cite{and:mat,Gon77,GRG80,KMSTTZ03,meh:tsa,PIA78,PerRei77,Pet57,W85,YY76}
cannot run beyond the $\Theta (\log \log n)$ expected time  for so
general as the smooth input distributions in \cite{and:mat}.
%
%
\vskip 0.05 cm
\noindent
{\bf Our result.} Under the most natural and even wider assumptions about
the random input keys that have been considered up to today
\cite{and:mat}, we manage to alleviate all lower bounds for the
dynamic predecessor search problem
\cite{Ajt88,Mil94,BF02,PatTho06,PatTho07}, by proving constant time
with high probability (whp), as $n$ grows large, thus improving over
all approaches in
\cite{and:mat,Gon77,GRG80,KMSTTZ03,meh:tsa,PIA78,PerRei77,Pet57,W85,YY76}
that achieved $\Theta (\log \log n)$ expected running time for such
wide random class of inputs. The fine details of our dynamic data
structure exhibit that it whp achieves {\em constant} predecessor
time working with only $O(n)$ ``short'' memory words of length at
most $O(\log n)$, just consuming an overall of $n^{1+ \delta}$
number of bits, for any constant $\delta
> 0$.
%
%
%
Also, our structure achieves $O(\log \log \log n)$ whp query time using only $n^{1+ \frac{1}{\log \log n}}= n^{1+o(1)}$ space, which is an exponential improvement over the pervious ${\tt IS}$ results.
Finally, we show whp $O(\log \log n)$ time using $O(n)$ space.
In this work we mainly demonstrated the asymptotic efficiency of the
data structure. The tuning of constants for practical
purposes is a matter of further work.

One could claim that just a simple trie with fan-out $n^\delta$ yields $O(1/\delta)$ time and linear space.
But note that this cannot (at least to the best of our knowledge) implement predecessor search
in constant time.
The reason is that, for the widest class of random input studied up to now, we managed to reduce whp the arbitrary universe size $|U|$ to the significantly smaller size $2^{C \log n}$, but, constant $C$ is $>1$. Hence the lower bounds in \cite{PatTho06} rule out any hope for linear space.
In fact, our way (Sect. \ref{Sec:Storing R into a FID}) can be viewed as an efficient VEB trie variant of \cite[Lem. 17]{PT06b}, which improves time to $O(\log (1/\delta))$.

Theorem \ref{Th:main result} below depicts the performance characteristics of our dynamic data structure (${\tt DS}$). These stem from the combination of many properties, formally proved in Section \ref{Sec:Properties of the upper static data}. It is more helpful to visit Section \ref{Sec:The data structure} first, illustrating the basic parts of {\tt DS}. It also presents the underlying intuition that governs our probabilistic analysis. Also, it suggests a way to follow the line of detailed proofs in Section \ref{Sec:Properties of the upper static data}.
\begin{theorem}
{\tt DS} initially stores a file $X$ of $n$ keys drawn from an
unknown $(n^{\gamma}, n^{\alpha})$-smooth input\footnote{This is superset from the class considered in \cite{and:mat}, see details in Def. \ref{Def:smooth-real} \& \ref{Def:smooth} and in Sect. \ref{subSec:whp each part of P contains log n input keys}} distribution $\Pr$, with $0 < \alpha < 1$ and constant $\gamma > 0$. For any positive constant $\delta > 0$ it has the following properties:
%
%

\noindent
{\bf 1.}
%
%
%
Requires $O(n^{1+\delta})$ space and time to built.
%
%

\noindent
{\bf 2.}
Supports $\Omega (n^{1+\delta})$ updates as: ${\tt Pred}(y, X)$, ${\tt Del}(y, X)$, ${\tt Insrt}(y,
X)$, ${\tt Memb}(y, X)$. We assume that  ${\tt Insrt}(y,
X)$ and ${\tt Del}(y, X)$ queries can occur in arbitrary order, only satisfying that $\Theta (n)$ keys are
stored per update step, while the rest queries are unrestricted.
%
%

\noindent
{\bf 3.}
Supports $O(1)$ query time with high probability (whp).
%
%

%
\noindent
{\bf 4.}
Supports $O(\log \log \log n)$ query time whp, using only $n^{1+ \frac{1}{\log \log n}}= n^{1+o(1)}$ space (built time).

\noindent
{\bf 5.}
Supports  $O(\log \log n)$  query time whp, using $O(n)$ space (built time).

\noindent
{\bf 6.}
Supports $O(1)$ whp amortized update time.
%

\noindent
{\bf 7.}
Supports $O(\sqrt{\frac{\log n}{\log \log n}})$ worst-case query time.
%
%
\label{Th:main result}
\end{theorem}
\begin{proof}
%
%
%

\noindent
{\bf 1.}
Sect. \ref{SubSec:Implementation of the idea} describes how ${\tt DS}$ is built. Sect. \ref{Sec:Storing R into a FID} proves in Th. \ref{Th:GORR09} the required time and space to construct the upper static {\tt DS} part. Finally, Sect. \ref{Sec:Implementing buckets of R as q* heaps} proves in Corol. \ref{Cor:Willard q* heap} the corresponding time and space required by the lower dynamic {\tt DS} part.
%
%

\noindent
{\bf 2.}
Sect. \ref{subSec:the number of stored keys per update step} describes in Property \ref{Prop:bounding the number of keys in X per step}  the sufficient conditions for performing $\Omega (n^{1+\delta})$ updates on the lower dynamic {\tt DS} part.
%
%

\noindent
{\bf 3.}
When handed an arbitrary query, we first proceed this query to the upper static {\tt DS} part, as Sect. \ref{Sec:Storing R into a FID} describes. Therein  Theorem \ref{Th:GORR09} proves that the handed query is driven whp in $O(1)$ time to the lower dynamic {\tt DS} part. Finally, Sect. \ref{Sec:Implementing buckets of R as q* heaps} describes how the handed query is implemented and Corollary \ref{Cor:Willard q* heap} proves that query is whp answered in $O(1)$ time.
%

%
\noindent
{\bf 4, 5.}
These are proved in Corollary \ref{Cor:logloglogn in o(n) bits}.

\noindent
{\bf 6.}
Sect. \ref{Sec:whp amortized time} proves the $O(1)$ whp amortized time.
%
%
%

\noindent
{\bf 7.}
Sect. \ref{Sec:amortized time} proves the $O(\sqrt{\frac{\log n}{\log \log n}})$ worst-case time.
%
%
\end{proof}
\section{Related work, smooth input, definitions}
\label{Sec:Previous work}
%
%
%
%
\noindent
{\bf Related word.}
The classical Van Emde Boas data structure~\cite{BKZ77} achieves $O(\log\ell)$  query time in exponential space, but the space can be reduced to linear when hashing is used~\cite{W83}.
%
%
Fredman and Willard~\cite{FW90,FreWil93} proposed the Fusion tree which achieves linear space with $O(\frac{\log n}{\log\ell})$ query time. When combining the fusion and Van Emde Boas bounds one gets a query time bound $O(\min(\log\ell,\sqrt{\log n}))$.
Capitalizing on this, Beame and Fich~\cite{BF02} were the first to dive below the van Emde Boas/Fusion tree barrier
, that is they improved the query time to $O(\min(\frac{\log\ell}{\log\log\ell},\sqrt{\frac{\log n}{\log\log n}}))$ while consuming quadratic in $n$ space.
%
%
The exponential tree~\cite{AT07} is a generic transformation which transforms any static polynomial space predecessor data structure with polynomial construction time into a dynamic linear space predecessor data structure with efficient updates. Combining the Beame and Fich data structure with exponential search tree produces a data structure with linear space and $O(\min(\log\log n\cdot \frac{\log\ell}{\log\log\ell},\sqrt{\frac{\log n}{\log\log n}}))$ query/update time. Finally Patrascu and Thorup~\cite{PatTho06} explore in depth every fine detail of the time versus space
interplay on a static set of $n$ keys of $\ell$ bits each.
%
%
In particular they have shown that for $\ell=O(\log n)$, that is for the case of a polynomial universe $|U|= 2^{O(\log n)}$, constant query time is possible and the optimal space is $O(n^{1+\delta})$, where $\delta$ is any positive constant.
In our case, $|U|$ can be arbitrarily larger than \cite{PatTho06} above, but we managed to reduce it to a polynomial one, via exploiting the smoothness of the random input.
Constant expected search time was proved for ${\tt IS}$ in \cite[Cor. 12]{and:mat} for random real input keys of a bounded probability density, formally for the restricted class of $(f_1= n, f_2= 1)$-smooth densities (Def. \ref{Def:smooth-real} below).
The work in \cite{KMSTTZ03} removed the dependence on $n$ for the expected search time by using finger search and obtaining $O(\log \log d)$ time, with $d$ the distance of the targed key from the key pointed by the finger.
The work in  \cite{KMSTTZ06} is the first that extends the analysis of ${\tt IS}$ to discrete keys.
The work \cite{DJP04} removes the independence of input keys.
%
%
We give the definition of smoothness for discrete input (see Def. \ref{Def:smooth-real} in App. \ref{app:A0} for the continuous case).
\begin{definition}[\cite{KMSTTZ06}]
An unknown {\em discrete} probability distribution $\Pr$ over the discrete keys  $x_1, \ldots, x_N$ of the universe $[a,~ b]$ is $(f_1, f_2)$-smooth\footnote{Intuitively, function $f_1$ partitions an arbitrary subinterval
$[c_1,c_3] \subseteq [a,b]$ into $f_1$ equal parts, each of length
$\frac{c_{3}-c_{1}}{f_1}= O(\frac{1}{f_1})$; that is, $f_1$ measures
how \emph{fine} is the partitioning of an arbitrary subinterval.
Function $f_2$ guarantees that no part, of the $f_1$ possible, gets
more probability mass than $\frac{\beta \cdot f_2}{n}$; that is,
$f_2$ measures the \emph{sparseness} of any subinterval
$[c_{2}-\frac{c_{3}-c_{1}}{f_1},c_2] \subseteq [c_1,c_3]$. The class
of $(f_{1},f_{2})$-smooth distributions (for appropriate choices of
$f_1$ and $f_2$) is a superset of both regular and uniform classes
of distributions, as well as of several non-uniform classes
\cite{and:mat,meh:tsa}. Actually, {\em any} probability distribution
is $(f_{1},\Theta(n))$-smooth, for a suitable choice of $\beta$.
}, if the exists a constant $\beta$ such that for all constants $a
\leq c_{1} < c_{2} < c_{3} \leq b$ and integers $n$,
%
for the conditional probability of a $\Pr$-random key $y$ to hold:
%
%
%
$$
P\left[c_{2}-\frac{c_{3}-c_{1}}{f_{1}(n)} \leq y \leq c_{2} | c_1 \leq
y \leq c_3\right]= \sum_{x_i= c_{2}-\frac{c_{3}-c_{1}}{f_{1}(n)}}^{ c_{2}}
{\Pr[c_{1},~ c_{3}](x_i)}
   \leq \beta \frac{ f_{2}(n)}{n}
$$
%
%
%
where $\Pr[c_{1},c_{3}](x_i)= 0$ for $x_i <c_1$ or $x_i > c_3$, and
$\Pr[c_{1},~c_{3}](x_i)=\Pr(x_i)/p$ for $x_i \in [c_1,~ c_3]$ where
$p=\sum_{x_i \in [c_1,~c_3]}\Pr(x_i)$.
\label{Def:smooth}
\end{definition}
%
%
%
%
\section{Data structure, underlying probabilistic intuition, proof
plan }
\label{Sec:The data structure}
%
%
%
%
\subsection{Implementation outline}
\label{SubSec:Implementation of the idea}
{\tt DS}  operates in two phases. The {\em preprocessing phase} samples $n$ random input
keys and fine tunes some {\tt DS} characteristics. The {\em
operational phase} consists of $\Omega(n^{1+\delta})$ update steps, each taking $O(1)$ time whp.
\vskip 0.05 cm
\noindent
{\bf Preprocessing phase:} $n$ random keys from the universe $U=
\{0, \ldots, 2^{|b|}\}$ are inserted to the dynamic file $X$, drawn
wrt an unknown $(f_1, f_2)$-smooth
distribution (Def. \ref{Def:smooth}). We order increasingly
the keys in $X$ as:
\begin{eqnarray}
\X= \{x_1, \ldots, x_n\}
\label{Def:X-ordered-increasingly}
\end{eqnarray}
We select $\rho-1 < n$ representative keys (order statistics):
\begin{eqnarray}
\Rb= \{r_0, \ldots, r_{\rho}\}
\label{Def:the representative keys in R}
\end{eqnarray}
out of the ordered file $\X$ in (\ref{Def:X-ordered-increasingly}):
 for $0<i< \rho$ (bucket) representative  key $r_i \in \Rb$ in (\ref{Def:the
representative keys in R}) is the key\footnote{We defer the
(technical) tuning of the value $\alpha (n) n = \Theta (\log n)$ in
the proof of Th. \ref{thm:red-dots}. Intuitively, the keys in $\Rb$ are the $\alpha (n) n$-apart keys in $\X$, except from $r_0 = 0$ and $r_{\rho}= 2^{b}$, being the endpoints of the universe $U$.} $x_{i\alpha (n)n} \in \X$
in (\ref{Def:X-ordered-increasingly}), with $r_0 = 0, r_{\rho}= 2^{b}$. We store an encoding
(indexing) of $\Rb$ into a VEB data structure (Sect. \ref{Sec:Storing R into a FID}), which consists the {\em static upper part} (${\tt sDS}$)  of our {\tt DS} (App. \ref{app:A}, Rem.
\ref{Rem:Static-FID-part-of-DS}).
Each pair of consecutive keys $r_i, r_{i+1} \in \Rb$ depicted in
(\ref{Def:the representative keys in R}) points to the $i$th bucket
with endpoints $[r_i, r_{i+1})$, for each $i \in [\rho-1]$. In turn,
the $i$th such bucket is implemented as $q^*$-heap \cite{Will92}. These  $q^*$-heaps consist the {\em dynamic lower part} (${\tt dDS}$) of our {\tt DS} (App. \ref{app:A}, Rem.
\ref{Rem:Dynamic-q*heaps-part-of-DS}), which can be
accessed (Sect. \ref{Sec:Implementing buckets of R as q* heaps}) by any query wrt an arbitrary key $y \in U$ in $O(1)$ time
only through the upper static ${\tt sDS}$ part of our ${\tt DS}$.

\noindent
{\em Properties:}
Observe that ${\tt sDS}$ may takes $O(n^{1+\delta})$ built time, but this is
absorbed by the $\Omega(n^{1+\delta})$ update steps of the
operational phase (App. \ref{app:A}, Rem. \ref{Rem:Static-FID-part-of-DS}), thus
safely inducing low whp amortized cost.
By a probabilistic argument (Th. \ref{Th:each part contains at most one key}), each representative $r_i \in \Rb$ whp can
be {\em uniquely} encoded with $C_1\log n$ bits, $C_1= O(1)$, yielding overall ${\tt sDS}$
space $O(n^{1+ \delta})$ for any constant $\delta >0$, while
performing ${\tt Pred} (y, \Rb)$ in $O(1)$ time for any key $y \in
U$ (Sect. \ref{Sec:Storing R into a FID}).
%

%
%
%
\vskip 0.05 cm
\noindent
%
{\bf Operational phase:} consists of
$\Omega(n^{1+\delta})$ queries (App. \ref{app:A}, Rem.
\ref{Rem:Static-FID-part-of-DS})
as: ${\tt Pred}(y, X)$, ${\tt Del}(y, X)$, ${\tt Insrt}(y, X)$,
${\tt Memb}(y, X)$ (Sect. \ref{Sec:Intro}).
Each query wrt key $y \in U$ is
landed in $O(1)$ time by the {\em static} ${\tt sDS}$
 of $\Rb$ defined in (\ref{Def:the representative keys in R})
towards to the appropriate {\em dynamic}  $q^*$-heap pointed by the
${\tt Pred} (y, \Rb) \equiv i_y$th bucket.

\noindent
{\em Properties:}
The only restriction is that ${\tt Insrt}(y, X)$ and ${\tt Del}(y,
X)$ queries can occur in an arbitrary (round-robin) round but must
satisfy that $\Theta (n)$ keys (Property \ref{Prop:bounding the number of keys in X per step}) are stored per update step, while the rest queries are unrestricted.
Since the ${\tt Pred} (y, \Rb) \equiv i_y$th such landed bucket is implemented as $q^*$-heap, each update wrt $y \in
U$ takes $O(1)$ time by the corresponding $q^*$-heap built-in
function. The only concern (App. \ref{app:A}, Rem. \ref{Rem:q*heap-load-remains-logn})
is to show that the $i_y$th bucket load whp remains $\Theta
(\log n)$ per update step in file $X$ (Th. \ref{Th:bits for reps} and Lem. \ref{thm:red-dots}).
\subsection{Bird's eye view of the idea}
\label{Sec:The idea of the data structure}
{\bf Uniform ``toy'' input.} For a toy warmup, let us assume the
uniform input distribution over $U$. We induce a partition of the
universe $U$ into $\rho < n$ {\em equal sized parts} or {\em
buckets}, in the sense that each bucket lies on equal number of keys
from $U$. More precisely, let $\Rb$ in (\ref{Def:the representative keys in R}) now deterministically partition as $[r_0, r_1) \cup [r_1, r_2) \cup \ldots \cup [r_{\rho-1}, r_{\rho}]$ the universe $U$, with
 $r_{0}= 0, r_{\rho}= 2^b= |U|$ and each part size is
$|[r_{i}, r_{i+1})|= \frac{|U|}{\rho}, \forall i \in [\rho - 1]$.
Now, observe that for any target key $y \in U$, we can locate within
$O(1)$ time the bucket $i_y \equiv [r_{i_y}, r_{i_y+1})$ that $y$
lies, in symbols $r_{i_y} \leq y  < r_{i_y+1}$. Just benefit from
the fact that {\em all} the buckets are equally sized and divide $y$
by the bucket size $\frac{|U|}{\rho}$. This yields bucket $i_y=
\lfloor y/\frac{|U|}{\rho}\rfloor$.
At this point, querying  ${\tt Pred}(y, X)$ in the original file has
been reduced in $O(1)$ time to querying ${\tt Pred}(y, X \cap
[r_{i_y}, r_{i_y+1}))$ in bucket $i_y$.
So it remains to compute ${\tt Pred}(y, X \cap [r_{i_y},
r_{i_y+1}))$ in $O(1)$ time.  This is achieved by implementing each
bucket $i$ as a $q^*$-heap \cite{Will92}, provided that each bucket
$i$ load $|X \cap [r_{i}, r_{i+1})|$ whp remains $\Theta(\log n)$
per update step, $\forall i \in [\rho -1]$.
By choosing $\rho= \frac{n}{\log n}$ it is folklore that each bucket
$i\equiv [r_{i}, r_{i+1})$ whp contains load $|X \cap [r_{i},
r_{i+1})|= \Theta(\log n)$ from a dynamic random file $X \subset U$
sized $n$.

\vskip 0.05 cm
\noindent
{\bf $\mu$-random input.} Let us now generalize the above toy, but
illustrative, assumption of uniform input. For each key $x \in U$,
assume an arbitrary {\em known} probability $\Pr[x]$ that key $x$ is
drawn to be inserted into file $X$. It is easy to exploit $\Pr[x]$
and partition the universe into (possibly $\rho
>> n$) buckets $[r_{i}, r_{i+1})$ so that whp each bucket load
$|X \cap [r_{i}, r_{i+1})|$ whp remains $\Theta(\log n)$. Just
compute for each bucket $i$ its endpoints $r_{i}, r_{i+1}$ that
satisfy $\sum_{x= r_{i}}^{r_{i+1}} \Pr[x] = \frac{\log n}{n}$. Whp
this imposes tight $\Theta(\log n)$ bucket load bound per update
step, as required for the $q^*$-heaps implementing the buckets to
remain $O(1)$ efficient.

However, now it is not so easy (as the division $\lfloor
y/\frac{|U|}{\rho}\rfloor$ above for the toy example of uniform
input) to locate the bucket $i_y \equiv [r_{i_y}, r_{i_y+1})$ that
query $y$ lies in. Because now each bucket $i$ has its own size
$|[r_{i}, r_{i+1})|$ possibly $\neq \frac{|U|}{\rho}$. The subtle
reason is that the bucket size $|[r_{i}, r_{i+1})|$ depends on the
accumulated probability mass $\sum_{x= r_{i}}^{r_{i+1}} \Pr[x]$ of
all the keys $x$ in the bucket $[r_{i}, r_{i+1})$. Intuitively, a
high valued probability function $\Pr[x]$ for the $x$'s in $[r_{i},
r_{i+1})$ would induce very narrow endpoints $r_{i}, r_{i+1}$, while
a low valued one would significantly widen the corresponding bucket
interval.
Here our rescue comes from storing all the bucket endpoints $\Rb$ in (\ref{Def:the representative keys in R}) into the ${\tt sDS}$ which essentially is a VEB structure exactly as described in \cite[Lem. 17]{PT06b}. Similarly as above, to determine the bucket
$[r_{i_y}, r_{i_y+1})$ that $y$ lies in, it reduces to compute the
predecessor of $y$ in the ${\tt sDS}$ of $\Rb$, which takes $O(1)$ time.

But, it remains the important concern of total ${\tt sDS}$ space to be
explained. It is critical to cheaply encode each bucket endpoint
$r_i \in \Rb$ in (\ref{Def:the representative keys in R}) by consuming the fewest bits. Roughly, the bits per
$r_i \in \Rb$  must not exceed $O(\log n)$, otherwise the total ${\tt sDS}$
space would be too excessive.
Observe that the required bits per  endpoint $r_i \in \Rb$, as well
as the total number $\rho$ of endpoints in $\Rb$, are critically
related to the smoothness of distribution function $\Pr$ (Def. \ref{Def:smooth}). Once more,
if $\Pr$ is too ``picky'', inducing some dense bucket $[r_{i},
r_{i+1})$, then the corresponding endpoints $r_{i}, r_{i+1}$  would
be too close to each other. Thus to encode each such endpoint
(amongst possibly $\rho >> n$ many ones) would sacrifice
$\omega(\log n)$ bits, yielding an overall space explosion of the
${\tt sDS}$.
We prove probabilistically (Th. \ref{Th:each part contains at most one key}) that such a bad scenario of highly dense
endpoints is quite unlikely when considering any smooth input
distribution in the class of \cite{and:mat}.
In this way, the total space is whp   $O(n^{1+ \delta})$ for any
constant $\delta >0$.
\vskip 0.05 cm
\noindent
{\bf Unknown smooth input as in \cite{and:mat}.} The  morals
discussed above will help us to tackle our final unresolved issue.
As it is common in real world applications, here the input
probability $\Pr$ is completely {\em unknown} to us, except that,
$\Pr$ belongs to the class (Def. \ref{Def:smooth}) of
smooth discrete distributions, which are more general (as Eq. (\ref{Eq:f1-f2-sparce partition}) describes) than the ones previously studied in \cite{and:mat}.
This causes a great obstacle for maintaining whp $\Theta(\log n)$
load per bucket being implemented as a $q^*$-heap: now $\mu$ is {\em
unknown}, so it is impossible to compute consecutive buckets
$[r_{i}, r_{i+1})$ as to conveniently hold $\sum_{x=
r_{i}}^{r_{i+1}} \Pr[x]= \frac{\log n}{n}$ per bucket $i$.

A random b2b game will help us to identify some ``well behaved''
endpoints $\Rb$ in (\ref{Def:the representative keys in R}) guaranteeing that each
of the $\rho-1$ consecutive buckets accumulates $\Theta
\left(\frac{\log n}{n}\right)$ probability mass (Th. \ref{Th:bits for reps} and Lemma \ref{thm:red-dots}).
Notice that here the endpoints in $\Rb$ are obtained
probabilistically, by exploiting the strong tail bounds induced by a
random game, a striking difference from the above examples, where
$\Rb$ was obtained deterministically by computing the summand of
probability $\Pr$ per bucket.
Intuitively, at the initialization phase of our data structure, the
set $X$ is formed by $n$ random inserted keys. Consider an arbitrary
bucket $I \subset U$ that (during the preprocessing phase, Sect.
\ref{SubSec:Implementation of the idea}) contains load $|X \cap I|=
\Theta (\log n)$ from $X$.  By the strong Binomial tail bounds, the
probability that bucket $I$ gets $\Theta (\log n)$ load, assuming
that its probability mass $\sum_{x \in I} \Pr[x]$ over the keys in
$I$ was either $o \left(\frac{\log n}{n}\right)$ or $\omega
\left(\frac{\log n}{n}\right)$, is exponentially small (Th. \ref{thm:red-dots-1}).
In words, the bad event that bucket $I$ gets load that deviates too
much its probability mass times $n$, is extremely unlikely.
The possible consecutive such buckets $I \subset U$ containing load
$|X \cap I|= \Theta (\log n)$ from $X$ are bounded by $n$. It
follows by the union bound that whp no bad event  occurs at any such
consecutive bucket.
Also, it follows that this will be true during each of the
subsequent $rn$ update operations on $X$, with $r= O(1)$.
%
%
%
%
%
%
%
%
%
\subsection{The proof plan}
\label{Sec:Proof plan}
%
%
%
%
\noindent
{\bf Step 1.}
In Section \ref{Sec:partitioning U into equal parts} we define a
partition $\P$ of the universe $U$ into $|{\cal P}|= n^{C_1= O(1)}$
{\em equal sized parts}, that is each part lies on the same number
$d_{{\cal P}}$ of elements of $U$.
Thus, each key $y \in U$ can be encoded into a single part of ${\cal
P}$, by applying a simple division $\frac{y}{d_{{\cal P}}}$, using
at most $\log \left(|{\cal P}|\right)= C_1 \log n$ bits. In words,
${\cal P}$ is a ``cheap'' way to index the keys of the universe $U$
with at most $C_1 \log n$ bits, which are enough for our ``high
probability'' purposes.

\noindent
{\bf Step 2.}
In Sections \ref{subSec:whp each part of P contains log n input
keys}, \ref{SubSec:keys in R are sufficiently apart} and
\ref{SubSec:P uniquely encodes R} the goal is to show that whp
${\cal P}$ {\em uniquely} encodes (indexes) each representative key
$r_i \in \Rb$ defined in (\ref{Def:the representative keys in R}), when $\alpha (n)n = C_4 \log n$ with $C_4= \Theta(1)$. It suffices to show that whp in each part of ${\cal
P}$ it is possible to lie at most one element from $\Rb$, in
symbols, $\forall p \in {\cal P}$ it holds $|p \cap \Rb| \leq 1$,
which is proved in Theorem \ref{Th:each part contains at most one
key}, an immediate consequence of Theorem \ref{Th:bits for reps} and Lemma
\ref{thm:red-dots}. In particular, Theorem \ref{Th:bits for reps}
shows that each (index) part $p \in {\cal P}$ whp contains very few
keys (say $|p \cap X| < C_{\P}\log n$) from the dynamic file $X$ (so
from $\Rb \subset X$ as well).
In turn, Lemma \ref{thm:red-dots} shows that whp each pair of
consecutive keys $r_i, r_{i+1} \in \Rb$ are sufficiently away or
apart, in symbols $|[r_i, r_{i+1}) \cap X |> C_3 \log n$, for a
sufficiently large constant $C_3 > C_{\P}$.
We conclude that Theorem \ref{Th:each part contains at most one key}
is a consequence of the fact that too many keys from the dynamic
file $X$ lie between any pair of consecutive but ``wide enough''
endpoints $r_i, r_{i+1} \in \Rb$, ruling out -in turn- any hope for
$r_i, r_{i+1}$ to lie in the same ``narrow enough'' part $p \in
{\cal P}$.
%
%

%
\noindent
{\bf Step 3.}
In Section \ref{Sec:Storing R into a FID} we index by ${\cal P}$ an
arbitrary query $y \in U$ and compute in $O(1)$ time its predecessor
$r_{i_y} \in \Rb$ using $n^{1+\delta}$ space, for any constant
$\delta >0$. We show this by uniquely indexing by ${\cal P}$ each $r_i \in
\Rb$ within $C_1 \log n$ bits and store the so-indexed $\Rb$ into our ${\tt sDS}$ (a static VEB as  \cite[Lem. 17]{PT06b}).
%

%
\noindent
{\bf Step 4.}
Finally, when handed $y$'s  predecessor $r_{i_y} \in \Rb$ (by {\bf
Step 3} above) in $O(1)$ time, in Section \ref{Sec:Implementing buckets of
R as q* heaps} we compute in $O(1)$ time $y$'s predecessor in the
dynamic file $X$, which is the goal of this paper. This is achieved
in $O(1)$ time in the corresponding $q^*$-heap that implements the
bucket of $r_{i_y} \in \Rb$, provided that $\Theta (\log n)$
elements from $X$ are stored in this bucket (Th. \ref{Th:bits for reps} and Lemma \ref{thm:red-dots}).
\noindent
{\bf Step 4.}
Finally, the corresponding amortized times are proved in Sections \ref{Sec:whp amortized time} and \ref{Sec:amortized time}.
\section{The proof details}
\label{Sec:Properties of the upper static data}
\subsection{A $n^{C_1= O(1)}$ partition ${\cal P}$ of the universe $U= [0, \ldots, 2^{b}]$}
\label{Sec:partitioning U into equal parts}
We show that each representative key in $\Rb$ is uniquely (and
cheaply wrt bits) indexed by a partition  $\P$ of universe $U$, in
the sense that any two distinct keys of $\Rb$ are mapped via $O(\log
n)$ bits to distinct parts of ${\cal P}$.
Let a  partition ${\cal P}$ of the universe $U= [0, \ldots, 2^{|b|}]$
into $|{\cal P}|= n^{C_1= O(1)}$ {\em equal sized parts}, that is,
each
part lies on the same number $d_{{\cal P}}$ of elements of $U$.
Constant $C_1$ depends only on the smoothness (Def. \ref{Def:smooth} and (\ref{Eq:f1-f2-sparce partition})) parameters $f_1, f_2$.
Each key $y \in U$ can be mapped into a single part of ${\cal P}$,
by applying a simple division $\frac{y}{d_{{\cal P}}}$. ${\cal P}$
is a ``cheap'' way to index each key in $U$ within $\log
\left(|{\cal P}|\right)= C_1 \log n$ bits, $C_1= O(1)$.
The goal is to show that $\P$ uniquely indexes the representative
keys of $\Rb$ by combining subsections \ref{subSec:whp each part of
P contains log n input keys} and \ref{SubSec:keys in R are
sufficiently apart}.
\subsection{Randomness invariance wrt the parts $\in \P$ per update step}
\label{Sec:Randomness invariance per update step}
The $n$ random keys inserted at the preprocessing phase (Sect.
\ref{SubSec:Implementation of the idea}) of $n$ steps and each
inserted key during the operational phase of $\Omega (n^{1+\delta})$
steps, are drawn wrt an {\em unknown} but smooth (Def. \ref{Def:smooth}) input distribution $\Pr$ over
$U$.
Furthermore, each deleted key during the operational phase is
selected uniformly at random from the currently stored keys.
In Section \ref{Sec:partitioning U into equal parts} we defined the
{\em deterministic} partition $\P$ of the universe $U= [0, \ldots,
2^{b}]$ into equal parts. Note that $\P$ is deterministic, thus, no
part $\in \P$ is constructed or biased by exploiting the random
sample. That is, no information has been exposed from any observed
random key as to facilitate (the length or position) of any part of
$\P$.
From such $\Pr$-random insertions and uniform deletions the work in
\cite{Knu77} implies that:
\begin{prop}
During each step of the preprocessing and operational phase the keys
stored are $\Pr$-randomly distributed per deterministic part of $\P$.
\label{Prop:randomness invariance wrt parts in P}
\end{prop}
\subsection{$\Theta (n)$ keys are stored per update step of the operational phase}
\label{subSec:the number of stored keys per update step}
Recall that during the operational phase (Sect. \ref{SubSec:Implementation of the idea}) queries ${\tt Insrt}(y, X)$ and ${\tt Del}(y, X)$ occur in an arbitrary (round robin) order and frequency, but must satisfy that
overall $\Theta (n)$ keys are stored in the dynamic file $X$ per update step. On the other hand
queries  ${\tt Pred}(y, X)$, ${\tt Memb}(y, X)$ are unrestricted. We conclude:
\begin{prop}
During each step $t= 1, \ldots, \Omega(n^{1+\delta})$, of the operational phase, the number $n_t$ of the keys
stored in file $X$ satisfy $n \leq C_{\min} n \leq n_t < C_{\max} n$, with $n$ the number of keys in $X$ at the end of preprocessing, and $C_{\min}, C_{\max}$ constants.
\label{Prop:bounding the number of keys in X per step}
\end{prop}
\subsection{Whp each part of ${\cal P}$ contains $\leq C_{\P} \log n$ keys from file $X$, with $C_{\P}= O(1)$}
\label{subSec:whp each part of P contains log n input keys}
The most general class of unknown $(f_1,f_2)$-smooth input distributions, so that $O(\log \log n)$ expected search times is achievable by ${\tt IS}$, was defined in \cite{and:mat} with $f_1, f_2$ parameters as:
\begin{eqnarray}
f_1(n)=\frac{n}{\log^{1+\epsilon} \log{n}}
~\mbox{and}~
f_2(n)= n^{\alpha}
%
%
\end{eqnarray}
with constants $0 < \alpha < 1$ and $\epsilon > 0$.
However, in our case, we will define an even more general class of smooth distributions:
\begin{eqnarray}
f_1(n)= n^{\gamma}
~\mbox{and}~
f_2(n)= n^{\alpha}
\label{Eq:f1-f2-sparce partition}
\end{eqnarray}
with constants $0 < \alpha < 1$ and $\gamma > 0$.
\begin{theorem}
\label{Th:bits for reps}
Let $n$ keys being inserted to file $X$ at the {\em Preprocessing phase}, obeying an unknown $(f_1, f_2)$-smooth
distribution from the general class considered in (\ref{Eq:f1-f2-sparce partition}). We can compute
constant $C_1= C_1(f_1, f_2)$, such that if $|{\cal P}|= n^{C_1}$
whp each part $p \in {\cal P}$ contains  $|p \cap X|< C_{\P}\log n$
random input keys during each update operation of the {\em
Operational phase} (Sect. \ref{SubSec:Implementation of the idea}),
with $C_{\P}= O(1)$.
\end{theorem}
\begin{proof}
See App. \ref{app:B}.
\end{proof}
\subsection{Whp the keys in  $\Rb$ are $> C_3 \log n$ apart, with constant $C_3 > C_{\P}$}
\label{SubSec:keys in R are sufficiently apart}
%
%
\subsubsection{Preprocessing (Sect. \ref{SubSec:Implementation of the idea}) partitions $U$ into $\rho$ parts of $\Theta (\frac{\log n}{n})$ probability mass.}
\label{SubSec:probability mass of each PR part}
Recall from the preprocessing phase (Sect. \ref{SubSec:Implementation of the idea}) the ordered file $\X$ in (\ref{Def:X-ordered-increasingly}) and the $\alpha (n) n$-apart representative keys in $\Rb$  defined in (\ref{Def:the representative keys in R}). We set $\alpha (n) n= C_4 \log n$, with $C_4= \Theta (1)$, which implies that key $r_i \in \Rb$ in (\ref{Def:the representative keys in R}) is the key $x_{(i+1)\alpha (n)n} \equiv x_{(i+1)C_4 \log n} \in \X$ in (\ref{Def:X-ordered-increasingly}).
In this way, the $\rho= \lfloor n/(C_4 \log n) \rfloor$ keys in $\Rb$ induce a partition $\P_{\Rb}$ of the universe $U$.

Observe that, opposed to partition $\P$ (Sect. \ref{Sec:partitioning U into equal parts}), the parts of $\P_{\Rb}$ are not necessarily equal. Also, opposed to partition $\P$, the endpoints of each part of $\P_{\Rb}$ are not deterministic, instead these are sampled random keys wrt the unknown input distribution $\Pr$. As a negative  consequence, Property \ref{Prop:randomness invariance wrt parts in P} is not further extendable to the corresponding parts of $\P_{\Rb}$, unless distribution $\Pr$ was continuous.
\begin{theorem}
\label{thm:red-dots-1}
Whp all parts of $\P_{\Rb}$ are spread on corresponding  subintervals of the
universe $U$ with probability mass $\Theta (\alpha (n))= \Theta (\frac{C_4 \log n}{n})$ where $C_4= \Theta (1)$.
\end{theorem}
\begin{proof}
See App. \ref{app:C}.
\end{proof}
\subsubsection{Operational phase: insert/delete balls into $\Theta (\log n)$ loaded bins.}
\label{SubSec:balls to bins}

\noindent
Recall (Sect. \ref{SubSec:probability mass of each PR part}) the partition $\PR$, induced by the keys in $\Rb$ defined in (\ref{Def:the representative keys in R}) at the preprocessing phase (Sect. \ref{SubSec:Implementation of the idea}). We interpret each $\PR$ part as a {\em bin} and each key in file $X$ as a {\em ball}.
Theorem \ref{thm:red-dots-1} implies that during the subsequent operational phase (Sect. \ref{SubSec:Implementation of the idea}) each {\em new} inserted ball into $X$  will land to any given bin with probability $\Theta (\frac{C_4 \log n}{n})$.
The existing $n$ balls, being inserted into file $X$ until the end of operational phase, are considered as {\em old}.
Property \ref{Prop:bounding the number of keys in X per step} implies that per update step $t= 1, \ldots, \Omega(n^{1+\delta})$, there are at most $n_t < C_{\max} n$ new balls thrown into the $\rho$ bins. Thus the load of each bin, wrt to the new balls, is governed by a Binomial random variable with $\Theta (C_4 \log n)$ expectation.  Its folklore to use its strong tail bounds and prove (each bin also contains $C_4 \log n $ old balls) that:
\begin{lemma}
We can compute  $C^*= O(1)> C_4:$ whp no bin induced by $\PR$  gets load $> C^* \log n$.
\label{lem:upper bound of bins load}
\end{lemma}
It remains to show that  per update step $t= 1, \ldots, \Omega(n^{1+\delta})$, whp there are at least $C_{\P} n$ new/old balls into each of the $\rho$ bins.
Trying to work as in the proof of the  upper bound in Lemma \ref{lem:upper bound of bins load}, a technical difficulty is that the distribution of the $n$ old balls is not the same as the new ones.
A way out is to {\em  remove} all the $n$ old balls (their sole purpose was to expose information from $\Pr$ and construct partition $\PR$ in a way that Th. \ref{thm:red-dots-1} holds) and insert again\footnote{It only induces an $O(n)$ time overhead to the preprocessing described in Sect. \ref{SubSec:Implementation of the idea}.} $n$ $\Pr$-random balls.
However, at the end of these $n$ random insertions, each $\PR$ part does not contain exactly  $C_4 \log n$ balls, rather it is expected to contain $\Theta (C_4 \log n)$ balls.
But now it is beneficial for the analysis that in each update step $t$, all the $n_t \geq n$ (Property \ref{Prop:bounding the number of keys in X per step}) existing balls are identically distributed into the $\rho$ bins, obtaining the following lemma.
\begin{lemma}
Per update step $t= 1, \ldots, \Omega(n^{1+\delta})$, each of the $n_t$ balls currently stored into file $X$ appears into an arbitrary bin with probability $\Theta (\frac{C_4 \log n}{n})$.
\label{lem:random-invariance-balls-bins}
\end{lemma}
An immediate consequence of lemma \ref{lem:random-invariance-balls-bins} is the following:
\begin{lemma}
We can compute $C_{\PR}= O(1)$ with $C_{\P}< C_{\PR} < C_4:$ whp no bin induced by partition $\PR$  gets load  $< C_{\PR}\log n$.
\label{thm:red-dots}
\end{lemma}
\subsection{Whp $\P$ uniquely indexes each $r_i \in \Rb$ with  $C_1 \log n$ bits, $C_1= O(1)$}
\label{SubSec:P uniquely encodes R}
\begin{theorem}
At the {\em Preprocessing phase} (Sect. \ref{SubSec:Implementation
of the idea}) let $n$ input input keys, drawn from an unknown $(f_1,
f_2)$-smooth distribution, to be inserted into file $X$.
Then we can compute constant $C_1= C(f_1, f_2, \rho)$ such that
$|{\cal P}|= n^{C_1}$ such that, during each update operation of the
{\em Operational phase} (Sect. \ref{SubSec:Implementation of the
idea}), each part $p \in {\cal P}$ whp contains at most one key: $|p
\cap \Rb| \leq 1$ from the set $\Rb= \{r_1, r_2, \ldots, r_{\rho}\}$
of the representatives of $X$.
\label{Th:each part contains at most one key}
\end{theorem}
\begin{proof}
It is immediate by combining Theorem \ref{Th:bits for reps} (each
part $p \in \P$ contains $< C_{\P} \log n$ keys from file $X$) and
Lemma \ref{thm:red-dots} (each pair of $r_i, r_j \in \Rb$ is at
least $C_3 \log n$ apart, with constant $C_3 > C_{\P}$).
\end{proof}
%
%
\subsection{Storing  $\Rb$ into a $O(1)$ time predecessor static data structure  using $O(n^{1+\delta})$ space,
$\forall \delta > 0$}
\label{Sec:Storing R into a FID}
%
%
In Sect. \ref{Sec:partitioning U into equal parts} we showed that
each representative key in $\Rb$ is uniquely (and cheaply wrt bits)
indexed by a partition  $\P$ of universe $U$, in the sense of
Theorem \ref{Th:each part contains at most one key}: any two
distinct keys of $\Rb$ are mapped via $C_1 \log n$ bits to distinct
parts of ${\cal P}$.
Let $\widetilde{\Rb}$ this indexing of $\Rb$ by partition $\P$ and
store $\widetilde{\Rb}$ into a VEB data structure.
%
For any target key $y \in U$, indexed as $\widetilde{y}$ by $\P$, we
want to locate in $O(1)$ time the unique bucket $i_{\widetilde{y}}
\equiv [r_{i_{\widetilde{y}}}, r_{i_{\widetilde{y}}+1})$ that
$\widetilde{y}$ lies, in symbols $r_{i_{\widetilde{y}}} \leq
\widetilde{y} < r_{i_{\widetilde{y}}+1}$, which is identified by its
bucket endpoints $r_{i_{\widetilde{y}}}, r_{i_{\widetilde{y}}+1} \in
\widetilde{\Rb}$.
It is obvious that $r_{i_{\widetilde{y}}}\equiv {\tt
Pred}(\widetilde{y}, \widetilde{\Rb})$ and below we take advantage
of a VEB variant as described in \cite[Lem. 17]{PT06b}.
\begin{theorem}{\em \cite[Lem. 17]{PT06b}}
\label{Th:GORR09}
Let $\kappa$ be any positive integer.
Given any set $S\subseteq [0,m-1]$ of $n$ keys, we can build in $O(n2^{\kappa}\log m)$ time a static data structure which occupies $O(n2^{\kappa}\log m)$ bits of space and solves the static predecessor search problem in $O(\log(\frac{\log m-\log n}{\kappa}))$ time per query.
\end{theorem}
In our case and terminology, we map the characteristic function of
the set $\widetilde{\Rb}$ to a bitstring  $S$ with length $|S|= m$
that equals the number of possible values of the keys stored in
$\widetilde{\Rb}$.
To figure out how large  is, recall from Sect.
\ref{Sec:partitioning U into equal parts} that each key
$\widetilde{r_i} \in \widetilde{\Rb}$ is uniquely indexed by $\P$
with $C_1 \log n$ bits. Hence the number $m$ of possible values of
the keys in $\widetilde{\Rb}$ are $|S|= m= 2^{C_1 \log n}$. Thus by setting
$\kappa= \delta$ we can build the predecessor data structure built on the $\widetilde{\Rb}$
takes $O(n^{1+\delta})$ space and answers to queries in time $O(\log(\frac{\log m-\log n}{\kappa}))=O(\frac{C_1-1}{\delta})=O(1)$.
\begin{corollary}
\label{Cor:logloglogn in o(n) bits}
Setting $\kappa=\delta\log n$, the predecessor data structure built on the $\widetilde{\Rb}$
takes $O(n^{1+\delta})$ space and answers to queries in time $O(\log(\frac{\log m-\log n}{\kappa}))=O(\log(\frac{C_1-1}{\delta}))=O(1)$.
By setting $\kappa=\frac{\log n}{\log\log n}$, the data structure built on $\widetilde{\Rb}$ takes $O(n^{1+1/\log\log n})=O(n^{1+o(1)})$ space and answers to queries in time $O(\log(\frac{\log m-\log n}{\kappa}))=O(\log(\frac{(C_1-1)\log n}{\frac{\log n}{\log\log n}}))=O(\log\log\log n)$. Finally by setting $\kappa=1$, the data structure built on $\widetilde{\Rb}$ takes $O(n)$ space and answers to queries in time $O(\log(\frac{\log m-\log n}{\kappa}))=O(\log((C_1-1)\log n))=O(\log\log n)$
\end{corollary}
\subsection{Implementing as a $q^*$-heap each $i$th bucket $[r_i, r_{i+1})$ with endpoints in  $\Rb$ }
\label{Sec:Implementing buckets of R as q* heaps}

\begin{corollary}[Cor. 3.2, \cite{Will92}]
Assume that in a database of $n$ elements, we have available the use
of pre-computed tables of size $o(n)$. Then for sets of arbitrary
cardinality $M \leq n$, it is possible to have available variants of
$q^*$-heaps using $O(M)$ space that have a worst-case time of $O(1 +
\frac{\log M}{\log \log n})$ for doing member, predecessor, and rank
searches, and that support an amortized time $O(1 + \frac{\log
M}{\log \log n})$ for insertions and deletions.
\label{Cor:Willard q* heap}
\end{corollary}
To apply Corollary \ref{Cor:Willard q* heap} in our case, the
database is the dynamic file $X$ with $n$ elements, so for our
purposes it suffices to use pre-computed table of size $o(n)$.
Also, the sets of arbitrary cardinality $M$ in our case are the
$\rho$ buckets $[r_i, r_{i+1})$ implemented as $q^*$-heaps, so by Lemma \ref{lem:upper bound of bins load} here
$M= C_4 \log n$ with $C_4= O(1)$. It follows that the worst-case
time of each $q^*$-heap is $O(1 + \frac{\log M}{\log \log n})= O(1+
\frac{\log (C_4 \log n)}{\log \log n})= O(1)$ and the space per
$q^*$-heap is $O(M)= O(C_4 \log n)= O(\log n)$.
\subsection{$O(1)$ whp amortized time}
\label{Sec:whp amortized time}
The built time of ${\tt DS}$ is $O(n^{1+\delta})$, while each update step $t= 1, \ldots, \Omega (n^{1+\delta})$ takes $O(1)$ time whp. It is straightforward that the amortized whp time is $O(1)$.
\subsection{$O(\sqrt{\frac{\log n}{\log \log n}})$ worst-case time}
\label{Sec:amortized time}
In order to get worst-case guarantees and ensure correctness of the data structure irrespective of number of keys insrted in the data structure, we will use two dynamic predecessor search data structures~\cite{BF02} with worst case $O(\sqrt{\frac{\log n}{\log \log n}})$ query and update times and linear space, see App. \ref{app:D}.
\bibliographystyle{plain}
%
\newpage
\appendix
\section{Unknown continuous $(f_1, f_2)$-smooth distributions}
\label{app:A0}
\begin{definition}[\cite{and:mat,meh:tsa}]
Consider an unknown {\em continuous} probability distribution over
the interval $[a,~ b]$ with density function $\mu(x)= \mu[a, b](x)$.
Given two functions $f_{1}$ and $f_{2}$, then $\mu(x)= \mu[a,b](x)$
is {\em $(f_{1},f_{2})$-smooth}  if there
exists a constant $\beta$, such that for all $c_{1},c_{2},c_{3}$, $a
\leq c_{1} < c_{2} < c_{3} \leq b$, and all integers $n$, it holds
that
%
%
$$
\Pr[X \in [c_{2}-\frac{c_{3}-c_{1}}{f_{1}(n)}, c_{2}]~| c_1 \leq X
\leq c_3 ]= \int^{c_{2}}_{c_{2}-\frac{c_{3}-c_{1}}{f_{1}(n)}}
{\mu[c_{1},c_{3}](x)dx}
   \leq \frac{\beta f_{2}(n)}{n}
$$
where $\mu[c_{1},c_{3}](x)=0$ for $x<c_1$ or $x>c_3$, and
$\mu[c_{1},c_{3}](x)=\mu(x)/p$ for $c_1\le x\le c_3$ where
$p=\int_{c_1}^{c_3}\mu(x) dx$.
\label{Def:smooth-real}
\end{definition}
%

\section{More on the preprocessing and operational phase}
\label{app:A}
\begin{remark}
\label{Rem:Static-FID-part-of-DS}
All the $\rho < n$ keys of the set $\Rb$ defined in
(\ref{Def:the representative keys in R}) are stored -{\em once and
for all operational phase}- into the ${\tt sDS}$ (which essentially is a VEB structure exactly as described in \cite[Lem. 17]{PT06b}) and are not subject to
 any of the update steps: these keys are fixed. Their sole purpose is to
land in $O(1)$ time each query request to the dynamic part of the
{\tt DS} which is a sequence of buckets being implemented as
$q^*$-heaps.
The sorted file $\X$ in (\ref{Def:X-ordered-increasingly}) is
considered as given, similar as \cite{and:mat,meh:tsa,KMSTTZ03} for
real keys of infinite length and as \cite{KMSTTZ06} for discrete
keys. The $\rho$ sampled $\alpha(n)n$-apart keys in file $\Rb$ in
(\ref{Def:the representative keys in R}) is a similar preprocess as
\cite{and:mat,meh:tsa,KMSTTZ03} where from random file $X$ a
different sample was taken for their corresponding tuning reasons.
Note that our robust operational phase lasts as long as
$\Omega(n^{1+\delta})$ time compared to $\Theta (n)$ in
\cite{and:mat,meh:tsa,KMSTTZ03}. The purpose of this is to cancel
out the $O(n^{1+\delta})$ time to built the static FID and induce
low amortized cost. Recall the infamous static structure in
\cite[Th. 4.5]{BF02} requires $\Omega (n^2)$ time to built, but its
dynamic version, although it is possible to built in $O(n)$ time,
its running time is pumped up by a $\log \log n$ factor \cite[Cor.
4.6]{BF02}, ruling out any attempt to beat IS time.
Finally, although the built time in \cite{and:mat,meh:tsa} is
$O(n)$,  it assumes a {\em real}-RAM of infinite machine word
length, which is different from the {\em finite word} RAM considered
here. For example, in {\em real}-RAM each infinite real word
occupies just 1 memory cell and real worlds can be processed in
$O(1)$ time.
\end{remark}
\begin{remark}
\label{Rem:Dynamic-q*heaps-part-of-DS}
As soon as a given query wrt key $y \in U$ is driven by
the FID of $\Rb$ defined in (\ref{Def:the representative keys in R})
towards to the appropriate $i_y$th bucket being implemented as a
$q^*$-heap, the query is answered by the corresponding $q^*$-heap
built-in function in $O(1)$-time. That is, all the keys of file $X$
are stored in the buckets and only the content of each bucket is
subject to any updates.
\end{remark}
\begin{remark}
\label{Rem:q*heap-load-remains-logn}
This is a simple balls to bins (b2b) game which proves
that the load per bucket remains $\Theta (\log n)$ per update step:
just prove that each $q^*$-heap bucket lies on a subset of $U$ with
$\Theta \left( \frac{\log n}{n}\right)$ probability mass.
\end{remark}
\section{Proof of Theorem \ref{Th:bits for reps}}
\label{app:B}
%
%
%
%
%
\noindent
$\bullet$
First,
we bound by $n^{C_1}= n^{O(1)}$ the total number of equal
parts of ${\cal P}$, in a way that each part is expected to receive
$\leq \log n$ random input keys  per update step $t= 1, \ldots, \Omega(n^{1+\delta})$, of the operational phase (Sect. \ref{SubSec:Implementation of the idea}). In other words, we want each such tiny $\P$ part to contain $\leq \frac{\log n}{n}$ probability mass from distribution $\P$.
%
%

Below we construct $\P$ recursively.
Note that the endpoints of each such $\P$ part are obtained deterministically, depending only on the given parameters $f_1, f_2$ defined in (\ref{Eq:f1-f2-sparce partition}).
Thus, Property \ref{Prop:randomness invariance wrt parts in P} governs the statistics of each $\P$ part per update step $t= 1, \ldots, \Omega(n^{1+\delta})$.
Property \ref{Prop:bounding the number of keys in X per step} implies that per update step $t= 1, \ldots, \Omega(n^{1+\delta})$ of the operational phase, for the current number $n_t$ of stored keys in file $X$ it holds:
\begin{eqnarray}
n_t < C_{\max} n \equiv \nu
\label{Eq:Def nu}
\end{eqnarray}
Initially, $\P$ partitions the universe $U= [0, \ldots, 2^{|w|}]$ into $f_1(\nu)\geq f_1(n_t)$ equally sized parts, with $f_1$ as in  (\ref{Eq:f1-f2-sparce partition}) and $\nu$ as in (\ref{Eq:Def nu}).
Smoothness (Def. \ref{Def:smooth})  and Eq. (\ref{Eq:f1-f2-sparce partition})\&(\ref{Eq:Def nu}) yield that each such $\P$ part gets $\leq \beta \frac{f_2(n_t)}{n_t} \leq \beta \frac{f_2(\nu)}{n_t}$ probability mass ((\ref{Eq:f1-f2-sparce partition}) implies $f_2$ is increasing wrt $n$) per update step $t$, with $\beta$ a constant (Def. \ref{Def:smooth}) depending only on the particular characteristics of distribution $\Pr$.
Hence, during each update step $t$, each $\P$ part expectedly gets $\leq \beta \frac{ f_2(\nu)}{n_t}\times n_t= \beta
\nu^{\alpha}$ of the $n_t \leq \nu$ keys (by (\ref{Eq:Def nu})) currently stored in $X$.
More accurately, the number of input keys distributed on any such
$\P$ part has exponentially small probability of deviating from its
expectation, thus, also from getting higher than its upper bound $\beta \nu^{\alpha}$.
For simplicity and without loss of generality we will not take into
account constant $\beta$ (Def. \ref{Def:smooth}).
Subsequent deterministic partitioning is applied recursively within
each such $\P$ part, until we reach a sufficiently small $\P$ part with probability
mass as low as possible in order for its expectation to be $\leq \log n$ (recall $n$ is $|X|$ as operational phase starts, Sect. \ref{SubSec:Implementation of the idea}) per update step $t= 1, \ldots, \Omega(n^{1+\delta})$. Let $h$ be the number of such recursive partitions,
then, it suffices:
\begin{equation}
\label{eq:height} \nu^{\alpha^{h}} \leq \log n \Longrightarrow h \leq \frac{\ln
\left(\frac{\ln \ln(n)}{\ln (\nu)}\right)}{\ln \alpha}
<
 \frac{\ln
\left(\frac{\ln \ln(n)}{\ln (n)}\right)}{\ln \alpha}
= -
\log_{\alpha} (\ln(n))
\end{equation}
where the strict inequality follows from (\ref{Eq:Def nu}).
\begin{remark}
Note that (\ref{eq:height}) implies that each such tiny $\P$ part gets an arbitrary but $\leq \frac{\ln n}{n}$ probability mass $q(n)$ wrt the unknown $\Pr$ distribution, with $n= |X|$ at the start of operational phase.
\label{Rem:upper bound prob mass of tiny P part}
\end{remark}
Bellow we upper bound the total number of $\P$ parts induced until the $h$ recursions grow up to the value of (\ref{eq:height}).

\noindent In the 1-st partition of the universe $U$, the number of
parts is $f_1(\nu)=f_1(\nu^{\alpha^0})=
(\nu^{\alpha^0})^{\gamma}=\nu^{\alpha^0\cdot \gamma}$.
%
In the 2-nd recursive partition, each part will be further
partitioned into $f_1(\nu^{\alpha})= f_1(\nu^{\alpha^1})=
(\nu^{\alpha^1})^{\gamma}=\nu^{\alpha^1\cdot\gamma}$ subparts.
%
In general, in the $(i+1)$-th recursive partition the subrange will
be divided into $f_1(\nu^{\alpha^i})=
(\nu^{\alpha^i})^{\gamma}=\nu^{\alpha^i\cdot\gamma}$ subparts.
%
Taking into account Eq.~(\ref{eq:height}), in the final level of
recursive partition, the total number $|{\cal P}|$ of subparts is
\begin{eqnarray}
|{\cal P}| &=& \prod_{i=0}^{h}{f_1 \left(\nu^{\alpha^i}\right)}=
\prod_{i=0}^{h}(\nu^{\alpha^i})^\gamma
< \prod_{i=0}^{h}{\nu^{\alpha^i\cdot \gamma}}= \nu^{\sum_{i= 0}^{h}\alpha^i\cdot\gamma}=
\nu^{\gamma\frac{\alpha^{h+1} -1}{\alpha -1}}= e^{\gamma\frac{\alpha}{\alpha -
1}}\nu^{\frac{\gamma}{1-\alpha}}<\nu^{\frac{\gamma}{1-\alpha}}
%
\nonumber \\
&\Rightarrow & |\P| < C_{\max}^{\frac{\gamma}{1-\alpha}} n^{\frac{\gamma}{1-\alpha}}
\label{Eq: polynomial parts in P}
\end{eqnarray}
with $\nu$ as in (\ref{Eq:Def nu}) and $\alpha$ as in (\ref{Eq:f1-f2-sparce partition}). Hence, the total number of bits needed for each such final part of
${\cal P}$ is at most
\begin{equation}
\log\left(|{\cal P}|\right)= \frac{\gamma}{1-\alpha}\log n + \log \left( C_{\max}^{\frac{\gamma}{1-\alpha}} \right)
\leq C_1 \log n
\end{equation}
%

%
\noindent
$\bullet$
Finally, we show that the probability that, any $\P$ part amongst the
$n^{C_1}$ many  depicted by (\ref{Eq: polynomial parts in P})
receives $\geq C_{\P} \times \log n$ input keys, approaches to 0 as $n$
grows large, with $C_1$\&$C_{\P}= O(1)$.
Without loss of generality, we compute the probability that a final $\P$
part with arbitrary probability measure $q(n)$  (with $q(n)$ dominated by the bounds of by Rem. \ref{Rem:upper bound prob mass of tiny P part}) contains $\alpha (n) $ input keys.
This probability equals
\begin{eqnarray}
{n \choose \alpha(n) n} q(n)^{\alpha(n)n} (1-q(n))^{(1-\alpha(n))n}
\sim \left[ \left( \frac{q(n)}{\alpha(n)}\right)^{\alpha(n)}
\left(\frac{1-q(n)}{1-\alpha(n)} \right)^{1-\alpha(n)} \right]^n
\label{p1}
\end{eqnarray}
\begin{remark}
The expression on the right of (\ref{p1}) is asymptotically equal to the
expression on the left if we use Stirling's approximation $n! \sim
\left(\frac{n}{e}\right)^n \sqrt{2 \pi n}$ and ignore inverse
polynomial multiplicative terms. Expression (\ref{p1}) is a convex
function of two variables ($q(n)$ and $\alpha(n)$) and achieves its
maximum when $q(n)=\alpha(n)$.
\label{Rem:asymptotics-binomial-distribution}
\end{remark}
Note that the probability that a fixed $\P$ part gets load of keys from $X$ that deviates to any higher value than $\alpha
(n)n= (1+\delta)\log n= C_{\P}\log n$ from its expected load $\leq \log n$ (Rem. \ref{Rem:asymptotics-binomial-distribution}) is at most:
\begin{eqnarray}
n \cdot \left[ \left(
\frac{q(n)}{\alpha(n)}\right)^{\alpha(n)}
\left(\frac{1-q(n)}{1-\alpha(n)} \right)^{1-\alpha(n)} \right]^n
\rightarrow 0, n\rightarrow \infty
%
\end{eqnarray}
since all the higher load values than $\alpha
(n)n= (1+\delta)\log n= C_{\P}\log n$ are $\leq n$ many, in addition, each such higher (more deviating) value occurs with probability at most the one in (\ref{p1}).
In addition, there are $n^{O(1)}$ final
parts (by Eq. (\ref{Eq: polynomial parts in P})) in ${\cal P}$, so the union bound yields:
\begin{eqnarray}
n^{O(1)} \cdot \left[ \left(
\frac{q(n)}{\alpha(n)}\right)^{\alpha(n)}
\left(\frac{1-q(n)}{1-\alpha(n)} \right)^{1-\alpha(n)} \right]^n
\rightarrow 0, n\rightarrow \infty
%
\end{eqnarray}
vanishing probability that any tiny $\P$ part gets any higher load from $X$ than $\alpha (n)n= (1+\delta)\log n= C_{\P}\log n$.
%
\section{Proof of Theorem \ref{thm:red-dots-1}}
\label{app:C}
Consider the $i$th part $I_i$ of the universe $U$ defined by the partition $\P_{\Rb}$, $i \in [\rho]$.
Let $q_i(n)$ the total probability mass spread by distribution $\Pr$ over all keys of $U$ that lie in part $I_i$, that is, $q_i(n)= \sum_{y \in I_i} \Pr(y)$.
Now, let us assume the bad scenario that the probability mass $q_i(n)$ of this part
$I_i$ is either $q_i(n)= \omega (\alpha(n))$ or $q_i(n)= o(\alpha(n))$.
According to this bad scenario, part $I_i$ of $\P_{\Rb}$ has a given probability mass $q_i(n)=  \omega (\frac{C_4 \log n}{n})$ or $q_i(n)= o(\frac{C_4 \log n}{n})$ and
contains $C_4 \log n$ consecutive keys from $X$. But, the probability of this bad scenario wrt part $I_i$ of $\P_{\Rb}$ is governed by the Binomial distribution:
\begin{eqnarray}
{n \choose \alpha(n) n} q_i(n)^{\alpha(n)n} (1-q_i(n))^{(1-\alpha(n))n}
\sim \left[ \left( \frac{q_i(n)}{\alpha(n)}\right)^{\alpha(n)}
\left(\frac{1-q_i(n)}{1-\alpha(n)} \right)^{1-\alpha(n)} \right]^n \rightarrow 0, n \rightarrow \infty
\label{Eq:prob-a-keys-lie-probability-measure-q}
\end{eqnarray}
which  by Remark \ref{Rem:asymptotics-binomial-distribution} vanishes exponentially in $n$ for $q_i(n)= \omega (\alpha(n))$ or $q_i(n)= o(\alpha(n))$.
Note that there are $\rho < n$ possible realizations of such bad scenaria (possible blocks of $C_4 \log n$ consecutive keys) hence the union bound gives:
\begin{eqnarray}
n \times \left[\max_{i \in [\rho]} \left\{\left( \frac{q_i(n)}{\alpha(n)}\right)^{\alpha(n)}
\left(\frac{1-q_i(n)}{1-\alpha(n)} \right)^{1-\alpha(n)}\right\} \right]^n
\rightarrow 0,
%
\end{eqnarray}
which also vanishes when each $q_i(n)= \omega (\alpha(n))$ or $q_i(n)= o(\alpha(n))$.
Thus, whp all parts of $\P_{\Rb}$ are spread on corresponding  subintervals of the
universe $U$ with probability mass $\Theta (\alpha (n))= \Theta (\frac{C_4 \log n}{n})$.
%
\section{Proof of $O(\sqrt{\frac{\log n}{\log \log n}})$ worst-case time}
\label{app:D}
\begin{enumerate}
\item In the first predecessor data structure which we note by $B_1$ we initially insert all the numbers in interval $[1,\rho]$. Then  $B_1$ will maintain the set of non empty buckets during the operational phase. The role of $B_1$ is to ensure the correctness of predecessor queries when some buckets get empty.
\item The second predecessor data structure which we note by $B_2$ is initially empty. The role of $B_2$ is to store all the overflowing elements which could not be stored in the $q^*$-heaps corresponding to their buckets.
\end{enumerate}
\paragraph{handling overflows.}
In addition, we maintain an array of $\rho$ counters where each counter $c_i$ associated with bucket $i$ stores how many keys are stored inside that bucket and assume that the capacity of the $q^*$-heap associated with each bucket is $C=\Theta(\log n)$. Now at initialisation, all the buckets have initial load $\Theta(\log n)$ and all the keys of any bucket $i$ are stored in the corresponding $q^*$-heap. Then at operational phase the insertions/deletions of keys belonging to a given bucket are directed to the corresponding $q^*$-heaps unless it is overflown. More precisely when trying to insert a key $x$ into a bucket number $i$ and we observe that $c_i\geq C$, we instead insert the key $x$ in $B_2$. Symmetrically, when deleting a key $x$ from a bucket $i$ when $c_i \geq C$ proceeds as follows : If the key $x$ is found into the $q^*$-heap, we delete it from there and additionally look into $B_2$ for any key belonging to bucket $i$ and transfer it to $q^*$-heap of bucket $i$ (that is delete it from $B_2$ and insert it into the $q^*$-heap). If the key $x$ to be deleted is not found in the $q^*$-heap, we instead try to delete the key from $B_2$. By using this strategy we ensure that any insertion/deletion in any bucket takes at worst $O(\sqrt{\frac{\log n}{\log \log n}})$ time and still $O(1)$ time whp. Queries can also be affected by overflown buckets. When the predecessor of a key $x$ is to be searched in an overflown bucket $i$ (that is when a predecessor search lands into bucket $i$ with $c_i>C$) and the key is not found in the corresponding $q^*$-heap, then the key $x$ is searched also in $B_2$ in time $O(\sqrt{\frac{\log n}{\log \log n}})$. As the event of an overflowing bucket is expected to be very rare the performance of queries is still $O(1)$ time whp.
\paragraph{handling empty buckets.}
The data structure $B_1$ will help us handle a subtle problem which occurs when we have some empty buckets. Suppose that we have a non empty bucket $i$ followed by a range of empty buckets $[i+1,j]$. Then the answer to any predecessor search directed towards any bucket $k\in [i+1,j]$ should return the largest element in bucket $i$.
Thus in the highly unlikely case that a predecessor search lands in an empty bucket $k$ (which is checked by verifying that $c_k=0$ ) we will need to be able to efficiently compute the largest non empty bucket index $i$ such that $i<k$ and this can precisely be done by querying $B_1$ for the value $k$ which obviously will return $i$ as $B_1$ is a predecessor data structure which stores precisely the index of non empty buckets and $i$ is the largest non empty bucket index preceding $k$. This last step takes $O(\sqrt{\frac{\log \rho}{\log \log \rho}})=O(\sqrt{\frac{\log n}{\log \log n}})$ time.
What remains is to show how to maintain $B_1$. For that we only need to insert a bucket $i$ into $B_1$ whenever it gets non empty after it was empty or to delete it from $B_1$ whenever it gets empty after it was enon empty and those two events (a bucket becoming empty or a bucket becoming non empty) are expected to be rare enough that the time bound for any insertion/deletion remains $O(1)$ whp.
\end{document}